\documentclass[sn-mathphys,Numbered]{sn-jnl}
\usepackage{graphicx}
\usepackage{multirow}
\usepackage{amsmath,amssymb,amsfonts}
\usepackage{amsthm}
\newtheorem{construction}{Construction}
\newtheorem*{construction*}{Construction}
\usepackage{mathrsfs}
\usepackage[title]{appendix}
\usepackage{xcolor}
\usepackage{textcomp}
\usepackage{manyfoot}
\usepackage{booktabs}
\usepackage{algorithm}
\usepackage{algorithmicx}
\usepackage{algpseudocode}
\usepackage{listings}
\usepackage{color}
\usepackage{makecell}
\newtheorem{theorem}{Theorem}

\newtheorem{example}{Example}%
\newtheorem{remark}{Remark}%

\newtheorem{definition}{Definition}%
\newtheorem{corollary}{Corollary}

\begin{document}

\title[Article Title]{Constructions of Quantum $(r,\delta)$-LRCs from cyclic codes}

\author[1]{\fnm{Rajendra Prasad} \sur{Rajpurohit}}\email{rajendrarajpurohit20@gmail.com}

\author*[1]{\fnm{Maheshanand} \sur{Bhaintwal}}\email{maheshanand@ma.iitr.ac.in}

\affil[1]{\orgdiv{Department of Mathematics}, \orgname{Indian Institute of Technology Roorkee}, \city{Roorkee}, \postcode{247667}, \country{India}}

\abstract{Classical $(r,\delta)$ locally recoverable codes (LRCs) play a central role in distributed data storage systems as they enable an efficient recovery from erasures by accessing a small number of surviving symbols. Motivated by their prospective use in future quantum data storage and by recent theoretical progress on quantum locally recoverable codes (qLRCs), we investigate the construction of qLRCs from classical cyclic $(r,\delta)$-LRCs. Our approach identifies cyclic LRCs whose defining sets satisfy a dual-containing condition, allowing them to serve as valid CSS ingredients. We present three explicit families of $(r,\delta)$-qLRCs, two of which are optimal with respect to the quantum Singleton-like bound,  whenever the codes are pure, thereby providing optimal examples. Additionally, the codes presented in Constructions \ref{c2} and \ref{c3} have no bound on their lengths with respect to the field size required to obtain these codes. 
}




\maketitle

\section{Introduction}

Locally recoverable codes (LRCs) were introduced to address the challenge of efficient recovery in distributed storage systems, in the event of node failures, by enabling recovery of erased symbols through a small number of surviving symbols \cite{Gopalan}. A code symbol is said to have locality \(r\) if it can be reconstructed by accessing at most \(r\) other symbols, so that each coordinate belongs to a small repair group, called \emph{local set}, which supports efficient single-erasure recovery with usually low bandwidth than code dimension.
To cope with the situation when multiple simultaneous erasures affect a single local set, the concept of \((r,\delta)\)-locality was introduced \cite{Prakash2012}, in which every coordinate lies in a local set of length at most \(r+\delta-1\) and minimum distance at least \(\delta\),  where $\delta \ge 2$, so that up to \(\delta-1\) erasures in each local set are correctable.

A good amount of literature has been developed on classical LRCs in the past decade. A singleton-type bound for $(r,\delta)$-LRCs is given in \cite{Prakash2012}. Several optimal constructions achieving the Singleton-like bound for LRCs have been proposed, including the important work of Tamo and Barg \cite{Tamo} based on \emph{good polynomials} and subcodes of Reed–Solomon codes. In recent years, numerous constructions have been proposed using cyclic and quasi-cyclic codes \cite{GR, TIAG, Bchen2018, chenzitan}, algebraic geometry codes \cite{Barg2015AG}, and matrix-product codes \cite{MP2023}. Further developments include LRCs with availability \cite{Rawat2016}, hierarchical locality \cite{HL, zhang2020, luo}, and unequal locality \cite{SK}. Collectively, these results form a rich theory connecting algebraic coding techniques with practical distributed storage requirements.

Quantum error-correcting codes (QECCs) form the foundation of reliable quantum computation and quantum communication. They protect quantum information from decoherence, noise, and adversarial disturbances by encoding logical qubits into larger systems using stabilizer or subsystem constructions \cite{NielsenChuang2000}. Particularly, QECCs, which are constructed via the Calderbank-Shor-Steane (CSS) framework, have formed the foundation of modern quantum coding theory, bridging classical error correction with quantum error correction \cite{Calderbank1996CSS, Steane1996CSS}.

A recent focus in quantum coding theory concerns locality properties and is motivated by practical architectural constraints in near-term quantum devices. Quantum locally recoverable codes (qLRCs) extend the classical notion of LRCs to quantum settings by enabling recovery of erased qudits using only a bounded number of other qudits in a repair set \cite{GG2023}. This notion parallels the classical $(r,\delta)$-LRC framework while addressing the fundamentally different structure of quantum error correction. This local recovery capability is of interest both for future large-scale quantum memories and for near-term quantum error correction.

The study of qLRCs is still developing, but has seen several recent contributions. Golowich and Guruswami \cite{GG2023} introduced qLRCs and provided the first constructions using quantum Tamo–Barg codes, and proved a Singleton-like bound for qLRCs. 
Since then, many constructions of such codes have been provided in the literature \cite{cao2025optimal,luo2025bounds, qLRC, bu2025quantum, xie2025two, zhou2025optimal}. These developments highlight ongoing progress, but explicit qLRC families derived from cyclic LRCs remain limited.

In this work, we construct three new families of quantum $(r,\delta)$ locally recoverable codes derived from classical cyclic $(r,\delta)$ LRCs. We obtain conditions on the defining set of a cyclic code $\mathcal{C}$ that ensure that the dual code of $\mathcal{C}$ is contained in $\mathcal{C}$, allowing these classical cyclic LRCs to serve as valid CSS components. Our Constructions \ref{C1} and \ref{c3} yield explicit qLRCs, which achieve the quantum Singleton-like bound for locality, whenever pure, thus giving optimal examples. Additionally, codes from Constructions \ref{c2} and \ref{c3} have no bound imposed on their lengths with respect to the field size required to obtain these codes. Moreover, to the best of our knowledge, Construction \ref{c2} is the first construction of LRCs which satisfy the afformanionted condition without imposing the condition $q \equiv 1 \mod{r+1}$. 

In Section~\ref{sec2}, we present the fundamental definitions and background required for LRCs and qLRCs. Section~\ref{sec3} provides constructions of dual-containing cyclic LRCs of length \( n = q - 1 \) over \( \mathbb{F}_q \). In Section~\ref{sec4}, we investigate dual-containing LRCs with locality \( r \) and \((r,\delta)\), respectively, whose lengths is not bounded by the field size required for their construction. From these classical codes, we derive the corresponding qLRCs.

While this work was in progress, we learned very recently that Galindo et al. \cite{galindo2026} have independently presented constructions of qLRC using BCH codes with minimum distance $\delta$. However, the construction methodology and the resulting code parameters are different from those obtained in our work. 

The parameters of the qLRCs constructed in this paper are given in Table \ref{Table:I}. Note that the codes presented in Theorem \ref{6.Thm} and Theorem \ref{thm6.3} are distance-optimal when pure.
\begin{table}[h]
\centering
\footnotesize
\caption{Parameters of the constructed codes, where $L=r+\delta-1$ and the notations are as defined in the corresponding references. }
\label{Table:I}

\renewcommand{\arraystretch}{1.2}

\begin{tabular}{|c|c|c|c|}
\hline
\text{Reference} &
\text{Classical LRC Parameters} &
\text{qLRC Parameters} &
\text{Restrictions} 
\\
\hline

\text{Theorem \ref{6.Thm}}
&
$[q-1=gL,(g-1)r+(\delta-1),r+1]$
&
$[[gL,(g-2)(r-\delta+1),\ge r+1]]_q$
&
\( 2 \leq \delta \leq r \)

\\
\hline

\text{Theorem \ref{thm6.2}}
&
$[q^m-1=ML,M(L-o_L(q)),\ge 2]$
&
$[[ML,M(L-2o_L(q)),\ge 2]]_q$
&
$\delta=2$, $o_L(q$) is odd

\\
\hline

\text{Theorem \ref{thm6.3}}
&
$[q^m-1,Mr,\delta]$
&
$[[ML,M(r-\delta+1),\ge \delta]]_q$
&
\makecell{$q\equiv1\pmod L,$\\
$2\le\delta\le\lceil L/2\rceil$}

\\
\hline

\end{tabular}
\end{table}

\section{Preliminaries}\label{sec2}

A linear code $\mathcal{C}$ of length $n$ and dimension $k$ is a $k$-dimensional subspace of $\mathbb{F}_q^n$, and is referred to as an $[n,k]$ linear code over $\mathbb{F}_q$. The elements of $\mathcal{C}$ are known as \emph{codewords}. The dual code of $\mathcal{C}$, denoted by $\mathcal{C}^\perp$, is defined as $\mathcal{C}^\perp = \{x \in \mathbb{F}_q^n \mid x \cdot c = 0 \text{ for all } c \in \mathcal{C} \}$. It is well known that $\mathcal{C}^\perp$ is also a linear code with parameters $[n, n - k]$.

The \emph{Hamming distance} between two codewords is the number of positions in which they differ. The \emph{minimum distance} of $\mathcal{C}$ is the smallest Hamming distance between any two distinct codewords in $\mathcal{C}$. If the minimum distance of $\mathcal{C}$ is $d$, then $\mathcal{C}$ is called an $[n,k,d]$ code.

A linear code $\mathcal{C}$ of length $n$ over $\mathbb{F}_q$ is called cyclic if, for every codeword $(c_0,c_1,\ldots,c_{n-1}) \in \mathcal{C}$, its cyclic shift $(c_{n-1},c_0,\ldots,c_{n-2})$ also belongs to $\mathcal{C}$. By associating each codeword with a polynomial $c(x)=c_0+c_1x+\cdots+c_{n-1}x^{\,n-1}$, cyclic codes can be viewed as ideals in the quotient ring $\mathbb{F}_q[x]/ \langle x^n-1 \rangle$. Since this ring is principal, every cyclic code $\mathcal{C}$ is generated by a unique monic divisor $g(x)$ of $x^n-1$, called the generator polynomial of $\mathcal{C}$. Thus, $\mathcal{C}=\langle g(x)\rangle$, and dimension of $\mathcal{C}$ is $k=n-\deg g(x)$. An alternative description is given through the defining set, which consists of the exponents of the roots of $g(x)$ in some extension field of $\mathbb{F}_q$. This algebraic framework makes cyclic codes both theoretically appealing and practically useful. 

\begin{definition}[Defining Set] Let $\mathcal{C}$ be a cyclic code of length $n$ over the $\mathbb{F}_q$, and let $\alpha$ be a primitive $n$-th root of unity in some extension field of $\mathbb{F}_q$.  
If $g(x)$ is the generator polynomial of $\mathcal{C}$, then the \emph{defining set} $Z$ of $\mathcal{C}$ is defined by 
$$Z=\{\, i \in \{0,1,\ldots,n-1\} \mid g(\alpha^i) = 0 \,\}.$$

\end{definition}

\begin{definition}\label{BCH}
   Let \(\mathcal{C}\) be a cyclic code of length \(n\) over \(\mathbb{F}_q\), and let \(\gamma\) be a primitive \(n\)-th root of unity in some extension field of \(\mathbb{F}_q\). If the defining set of \(\mathcal{C}\) contains a set of consecutive powers
\[
\{\gamma^{b}, \gamma^{b+1}, \ldots, \gamma^{b+\delta-2}\}
\]
for some integer \(b\), then the minimum distance of \(\mathcal{C}\) satisfies
\[
d(\mathcal{C}) \ge \delta .
\]
\end{definition}

The following theorem provides a condition for a cyclic code to contain its dual.

\begin{theorem}\label{DC}
   Let \( \mathcal{C} \subseteq \mathbb{F}_q^n \) be a cyclic code of length \( n \) with  defining set \( {Z} \subseteq \mathbb{Z}_n \). Then, \( \mathcal{C} \) is a dual-containing code if and only if \({Z} \cap -{Z} = \varnothing \), where \(-{Z}\) is defined as 
   $-{Z}=\{n-i : i \in {Z}\}$.
\end{theorem}

\begin{definition}
Let $n$ be a positive integer and $q$ be an integer such that $\gcd(q, n) = 1$. For any integer $i$ such that $0 \le i \le n-1$, the $q$-cyclotomic coset of $i$ modulo $n$, denoted by $S_i$, is defined by
\[
S_i = \left\{i \cdot q^j \pmod{n} \mid j \in \mathbb{Z}_{\ge 0}\right\}
.\]
\end{definition}

\begin{definition}
Let $n \ge 2$ be a positive integer and let $q$ be an integer such that $\gcd(q,n)=1$.
The order of $q$ modulo $n$, denoted by $O_n(q)$, is the smallest
positive integer $t$ such that
\[
q^t \equiv 1 \pmod{n}.
\]
\end{definition}

\subsection{Locally recoverable codes}

 Let $\mathcal{C}$ be a linear $[n, k]$ code. We say that the $i$-th coordinate of $\mathcal{C}$ has locality $r$ if there exists a set $R \subseteq[n] \backslash\{i\},|R| \leqslant r$, such that across all codewords $c\in \mathcal{C}$, the value of the coordinate $c(i)$ is determined by the values of the coordinates $\{c(j)\} \mid  j \in R\}$. Equivalently, the $i^{th}$ coordinate has locality $r$ if the dual code $\mathcal{C}^{\perp}$ contains a codeword $c$ of Hamming weight at most $r+1$ such that the coordinate $i$ is in the support of $c$. 
 It is shown in \cite{Gopalan} that a linear code $\mathcal{C}$ with locality $r$ and minimum distance $d$ satisfies the upper bound
$$
d \leq n-k-\left\lceil\frac{k}{r}\right\rceil+2.
$$ 
This bound is called the Singleton-type bound for LRCs.

The codes with $(r,\delta)$ locality are a generalization of LRCs, and address the problem of recovering erased symbols in a code when more than one data symbol is erased from a local set.
\begin{definition}\label{definition 1}   
The $i^{th}$ coordinate of an $[n, k, d]$ linear code $\mathcal{C}$ is said to have $(r, \delta)$-locality if there exists a subset $R_{i} \subseteq[n]$ such that \\
(i) $i \in R_{i}$ and $\left|R_{i}\right| \leqslant r+\delta-1$; \\
(ii) The minimum Hamming distance of the punctured code $\mathcal{C} | _{R_{i}}$ obtained by deleting the code symbols $c_{j}, j \in[n] \backslash R_{i}$, is at least $\delta$. 
\end{definition}
The set $R_i$ is referred to as a \emph{local set} for the $i^{th}$ coordinate.

A linear code $\mathcal{C}$ is said to have all symbol $(r, \delta)$-locality and is called an $(r, \delta)$ LRC if all the $n$ coordinates have $(r, \delta)$-locality \cite{KPLK}.
It is well known  that the minimum distance $d$ of a code with $(r, \delta)$-locality satisfies the bound
\begin{equation} \label{d_LRC}
d \leq n-k+1-\left(\left\lceil\frac{k}{r}\right\rceil-1\right)(\delta-1).
 \ \ \ \  \end{equation} 
A linear code with $(r, \delta)$-locality is called distance-optimal if it achieves the above bound.

\subsection{Quantum stabilizer codes and qLRC}

Let $H$ be a Hilbert space of dimension $q$ over the field $\mathbb{C}$. The $n$-fold tensor product $H^{\otimes n}$ is a Hilbert space of dimension $q^{n}$ over $\mathbb{C}$. A quantum code with parameters $[[n,k,d_{Q}]]_{q}$ is a $q^{k}$-dimensional subspace $\mathcal{Q} \subseteq H^{\otimes n}$, where $d_{Q}$ denotes the minimum distance of the code.

Calderbank and Shor \cite{Calderbank1996CSS}, and Steane \cite{Steane1996CSS}, introduced a method for constructing quantum codes from classical linear codes. This method is known as the Calderbank--Shor--Steane (CSS) construction. It provides a way to obtain quantum stabilizer codes from classical codes that satisfy a suitable duality condition.

\begin{theorem}[CSS construction {\cite{Calderbank1996CSS,Steane1996CSS}}]\label{CSS}
Let $\mathcal{C}_{1}$ and $\mathcal{C}_{2}$ be two linear codes over $\mathbb{F}_{q}$ with parameters 
$[n,k_{1}]$ and $[n,k_{2}]$, respectively. If $\mathcal{C}_{2}^{\perp} \subseteq \mathcal{C}_{1}$, then there 
exists a quantum code with parameters $[[n,\,k_{1} + k_{2} - n,\,d_Q]]_{q}$, where
\[
d_Q = \min\{\, w_{H}(C_{1} \setminus C_{2}^{\perp}),\ 
                w_{H}(C_{2} \setminus C_{1}^{\perp}) \,\}.
\]
In particular, if $\mathcal{C}^{\perp} \subseteq \mathcal{C}$ for a linear code $\mathcal{C}$ of parameters 
$[n,k,d]$, then there exists a quantum stabilizer code $\mathcal{Q(C)}$ with parameters
\[
[[n,\,2k - n,\,d_Q]]_{q},
\qquad
d_Q = w_{H}(\mathcal{C} \setminus \mathcal{C}^{\perp}).
\]
\end{theorem}

A quantum code $\mathcal{Q(C)}$ is said to be \emph{pure} if its minimum distance $d_Q$ is equal to the minimum distance of the code $\mathcal{C}$.

The CSS construction provides a systematic way to obtain quantum stabilizer codes from classical linear codes that satisfy the dual-containing property. Using this construction, a stablizer code was presented in \cite{qLRC} which satisfies the property of $(r,\delta)$ quantum LRC. Their key results include deriving qLRCs using a classical LRCs. We first introduce the definition of $(r,\delta)$ quantum LRC as described in \cite{qLRC}.

\begin{definition}[Quantum $(r,\delta)$-locally recoverable codes]
Let $\mathcal{Q} \subset \mathbb{C}_q^{\otimes n}$ be an $[[n,K,d]]_q$ 
quantum code.
The $i$-th coordinate of $\mathcal{Q}$ is said to have $(r,\delta)$-locality
if there exists a subset
\[
R_i \subseteq [n], \qquad i \in R_i, \quad ~\mbox{and}~  \quad |R_i| \le r+\delta-1,
\]
such that the code $\mathcal{Q}$ can recover from \emph{any}
set of $\delta-1$ erasures occurring in the coordinate set $R_i$.

A quantum code $\mathcal{Q}$ is said to be a quantum $(r,\delta)$-LRC
if every coordinate has $(r,\delta)$-locality.

\end{definition}

More about the basic notions of quantum computation can be found in the references \cite{Gottesman1997,Calderbank1996CSS,NielsenChuang2000}.

\begin{theorem}\label{CqLRC}
Let $\mathcal{C} \subseteq \mathbb{F}_q^n$ be a linear code with parameters $[n, k, d]_q$. Suppose that $\mathcal{C}$ is a Euclidean dual-containing code, i.e. $\mathcal{C}^{\perp} \subseteq \mathcal{C}$, and that
$
\delta \leq d\left(\mathcal{C}^{\perp}\right) .
$

Define $Q(\mathcal{C})$ to be the stabilizer quantum code obtained from $\mathcal{C}$ via the standard $\mathrm{CSS}$ construction given in Theorem \ref{CSS}. Then:
\begin{enumerate}
    \item 
 $Q(\mathcal{C})$ has the parameters $[[n, \kappa= 2 k-n, d_Q\geq d]]_q$.

\item  $\mathcal{C}$ is a classical $(r, \delta )$-LRC if and only if $Q(\mathcal{C})$ is an $(r, \delta)$-qLRC.

\end{enumerate}

\end{theorem}

Moreover, if $Q(\mathcal{C})$ is a pure quantum code, then the parameters $[[n, 2 k-n, d_Q=d]]_q$ of $Q(\mathcal{C})$ satisfy the quantum Singleton-like bound \cite{qLRC}
$$
k+2 d_Q+2\left(\left\lceil\frac{n+k}{2 r}\right\rceil-1\right)(\delta-1) \leq n+2.$$
A pure $(r,\delta)$-qLRC is said to be optimal if it satisfies the above bound with equality.

\section{Dual-containing  cyclic  LRCs with \texorpdfstring{$n|(q-1)$}{n divides q-1} and quantum LRCs}\label{sec3}

In this section, we construct cyclic codes through the set of zeros of the generator polynomial, ensuring  that the code satisfies the condition of locally recoverability as well as the dual-containing property. First we give an important result which will be useful in our constructions.


\begin{theorem}\label{setA}\cite{FANG2020101650} Let $n$ be a positive integer such that $\gcd(n,q)=1$. Let $\gamma \in \mathbb{F}_{q^{s}}$ be a primitive $n$-th root of unity, where $\mathbb{F}_{q^{s}}$ is the splitting field of $x^{n}-1$ over $\mathbb{F}_{q}$. 
Let $r$ and $\delta$ be positive integers such that $2 \le \delta \le r$ and $L = r+\delta-1$ divides $n$. Set $M = \frac{n}{L}$.  

Let $\mathcal{C}$ be a cyclic code of length $n$ over $\mathbb{F}_{q}$ with complete defining set $Z$.  
Consider integers $l_{1} < l_{2} < \cdots < l_{\delta-1},$
forming an arithmetic progression with $\delta-1$ terms and common difference $b$, where $\gcd(b,n)=1$. Define the set
\[
\mathcal{A} = \{~jL+l_i :~~ 1 \le i \le \delta-1, ~0 \le j \le M-1\}.
\]

If $Z$ contains $\mathcal{A}$, then the cyclic code $\mathcal{C}$ has $(r,\delta)$-locality.

\end{theorem}

Using Theorem \ref{setA}, we present a construction of a dual-containing $(r,\delta)$-LRCs, which we can use as CSS components for constructing qLRCs.

\begin{construction}\label{C1}
Let $n=q-1$, and let $\alpha$ be a primitive $n$-th root of unity in $\mathbb{F}_q$. Let \( r \) and \( \delta \) be positive integers such that \( 2 \leq \delta \leq r \), \( L = r + \delta - 1 \), \( L \mid n \), and \( g = n/L \).

Define the set 

\[
A = \bigcup_{j=0}^{g-1} \{ jL + 1, jL + 2, \dots, jL + (\delta - 1) \} \subseteq \mathbb{Z}_n.
\]

For any fixed \( u \in [1, g-1] \) such that \( (g - 2u)L \pmod{n} \notin [0, 2r - 2] \), define a set 

\[
B = \{ uL + i : 0 \leq i \leq r - 1 \}. 
\]
Let $\mathcal{C}$ be the cyclic code with $Z = A \cup B$ as the defining set. 
\end{construction}

\begin{theorem}\label{6.Thm}
The linear code $\mathcal{C}$ given in Construction 1 above, satisfies the following properties:
\begin{enumerate}
    \item $\mathcal{C}$ is a distance-optimal cyclic dual-containing $(r,\delta)$-LRC over $\mathbb{F}_q$ with parameters $[g(r+\delta-1),(g-1)r+ (\delta-1),r+1]$.
    \item The minimum distances $d(\mathcal{C})$ and $d(\mathcal{C}^{\perp})$ of $\mathcal{C}$ and $C^\perp$, respectively,   satisfy
    \[
    d(\mathcal{C}) =  r+1, \qquad d(\mathcal{C}^\perp) \ge  r+1 \ge \delta.
    \]
    \item There exists an $(r,\delta)$-qLRC with parameters $[[n, 2 k-n, \ge r+1]]_q=[[g(r+\delta-1), (g-2)(r-\delta+1), \ge r+1]]_q$. 
\end{enumerate}

\end{theorem}
\begin{proof}
\begin{enumerate}
  
    \item 
Notice that the defining set $Z$ of $\mathcal{C}$ contains the set $\mathcal{A}$, with $l_i=i$, as described in Theorem \ref{setA}, hence the code $\mathcal{C}$ is an $(r,\delta)$-LRC.

Now, to show that $\mathcal{C}$ is dual-containing, we first show that $A \cap (-A) = \varnothing$.  
For any $x \in A$, we have $-x \pmod{n} = n-x$. Since $x = jL+t$ for some $j \in [0,g-1]$ and $t \in [1,\delta-1]$, it follows that  
\[
-x = n-(jL+t) = (g-j)L - t = sL - t
\]
for some $s \in [0,g-1]$.  

Now suppose $-x \in A$. Then, for some $1 \le t'\le \delta-1$ and $0\le k\le g-1$, we must have
\[
kL+t' \equiv sL-t \pmod{n}.
\] 
This congruence holds for certain $k$ and $s$ if and only if $t+t'$ is a multiple of $L = r+\delta-1$.  

However, the maximum possible value of $t+t'$ is $2(\delta-1)$, and since $\delta \le r$, we obtain
\[
t+t' < r+\delta - 1 = L,
\]
which implies that $t+t'$ cannot be a multiple of $L$. This is a contradiction. Therefore, no such $x$ exists in $A$, and hence
$A \cap (-A) = \varnothing.$

Next, we consider the intersection $A \cap (-B)$.  
Recall that $B = \{ uL + i_1 : 0 \le u \le g-1,\; 0 \le i_1 \le r-1 \}$, where $L = r+\delta-1$ and $n = Lg$.  

For any $y \in B$, we have $y = uL+i_1$ with $0 \le i_1 \le r-1$. Then
\[
-y \pmod{n} \in A \text{ if and only if } -(uL+i_1) \equiv jL+t \pmod{n}
\]
for some $j \in [0,g-1]$ and $t \in [1,\delta-1]$.  
This congruence is equivalent to
\[
(j+u)L + (i_1 + t) \equiv 0 \pmod{n}.
\]
Thus, for the congruence to hold, we must have $i_1 + t \equiv 0 \pmod{L}$, i.e., $i_1+t$ should be at least $L$.  

However, since $t \le \delta-1$ and $i_1 \le r-1$, we obtain
\[
i_1+t \le (r-1)+(\delta-1) = r+\delta-2 < r+\delta-1 = L,
\]
which implies that \(i_1 + t\) cannot be a multiple of \(L\). Consequently, no such \(y\) exists, and we conclude that
$A \cap (-B) = \varnothing.$

Finally, we show that $(-B) \cap B = \varnothing$. For any $y = uL + i_1\in B \cap -B$, there must exist integers $i_2 \in [0,r-1]$ such that  
\[
(g-u)L - i_1 \equiv uL + i_2 \pmod{n}.
\]
This congruence simplifies to  
\[
i_1 + i_2 \equiv (g - 2u)L \pmod{n}.
\]
The equation does not have a solution since $(g - 2u)L \pmod{n} \notin [0, 2r - 2]$. Therefore, we conclude that $B \cap -B = \varnothing.$

Combining all together, and by Theorem \ref{DC}, it follows that $\mathcal{C}$ is a dual-containing cyclic code. 

\vspace{2mm}

\item Note that the defining set $Z$ of $\mathcal{C}$ contains the set $\{uL,uL+1,uL+2, \cdots , uL+r-1\}$, comprising \( r \) consecutive elements. By BCH bound (\ref{BCH}), the minimum distance $d(\mathcal{C})$ of $\mathcal{C}$ satisfies $d(\mathcal{C})\ge r+1$.

It remains to show that $d(\mathcal{C})\le r+1$. Note that the set $Z$ contains $(g-1)(\delta-1)+r$ elements. The dimension $k$ of $\mathcal{C}$ is given by 
\begin{align*}
   k&=n-|Z|=g(r+\delta-1)-((g-1)(\delta-1)+r)
   \\ &= gr+g(\delta-1)-g(\delta-1)+(\delta-1)-r
   \\ &=(g-1)r+ (\delta-1).
\end{align*} 
Now, using the distance bound (\ref{d_LRC}), we obtain 
\begin{align*}    
d(\mathcal{C}) & \leq n-k+1-\left(\left\lceil\frac{k}{r}\right\rceil-1\right)(\delta-1)
\\ & \le (g-1)(\delta-1)+r)+1 - \left(\left\lceil\frac{(g-1)r+ (\delta-1)}{r}\right\rceil-1\right)(\delta-1) 
\\ & \le (g-1)(\delta-1)+r)+1 - \left(\left\lceil (g-1)+\frac{ (\delta-1)}{r}\right\rceil-1\right)(\delta-1) 
\\ & \le (g-1)(\delta-1)+r)+1 - (g-1)(\delta-1), \ \ \ \text{since } \delta-1 < r
\\ & \le r+1.
\end{align*}
Hence, $d(\mathcal{C}) = r+1$ and the code $\mathcal{C}$ is a distance-optimal $(r,\delta)$-LRC.

For $d(\mathcal{C}^\perp)$, we have $\mathcal{C}^\perp \subseteq \mathcal{C}$. This simply implies that $d(\mathcal{C}^\perp) \ge d(\mathcal{C}) = r+1 \ge \delta$ as we have $r \ge \delta$.

\vspace{2mm}

\item As the code $\mathcal{C}$ is an LRC with $(r,\delta)$ locality, and with  $d(C^\perp) \ge r \ge \delta$, the quantum code $Q(\mathcal{C})$ obtained by the CSS construction as claimed in Theorem \ref{CqLRC}, is an $(r,\delta)$-qLRC with parameters $[[g(r+\delta-1), (g-2)(r-\delta+1), \ge r+1]]_q$.

\end{enumerate}

\end{proof}

{\remark{Since the classical code obtained by Construction \ref{C1} is distance-optimal, the corresponding quantum code $Q(\mathcal{C})$ is distance-optimal whenever pure.  
}}

{\remark{

The construction of qLRCs from dual-containing LRCs was presented in \cite{luo2025bounds} for the case $\delta = 2$. In this setting, our construction differs from theirs in the following aspects. First, the construction in \cite{luo2025bounds} achieves minimum distance $l+1$, where $l+1 < r+1$ since $l<r$ is assumed. In contrast, our construction attains minimum distance $r+1$. Moreover, we remove the additional restriction $u+2l<r+2$ imposed in \cite{luo2025bounds}. To increase the minimum distance, the authors in \cite{luo2025bounds} include an additional $l$ roots in the first coset. Instead, we make use of the condition $(g-2u)L \pmod{n} \notin [0,\,2r-2]$
and place these roots in the $u$th coset. This approach allows us to increase the minimum distance to $r+1$. 
}}

We present two examples: one for the case when the obtained code is pure, and the other for the case when it is not.
\begin{example}
Let $n=2^4-1=15$, $r=4$ and $\delta=2$. We get $L=r+\delta-1=5$ and $g=3$. We take $u=2$ and verify that 
$(g-2u)L=(3-4)\cdot5=-5 \equiv 10 \pmod{15}$ and 
$10
\notin
[0,6]$.
The set $A$ in this case is
$A= \cup_{j=0}^{2} \{5j+1\}=\{1,6,11\}$ and the set $B=\{10,11,12,13\}$.

The set $Z= \{1,6,10,11,12,13\}$ is the defining set for the code $\mathcal{C}$.
Using MAGMA calculator \cite{MR1484478}, we verified that the code $Q(\mathcal{C})$ is pure distance-optimal $(r=4,\delta=2)$-qLRC with parameters $[[15, 3, 5]]_{q=2^4}$.
\end{example}

\begin{example}
Let $n=13-1=12$, $r=3$ and $\delta=2$. We get $L=r+\delta-1=4$ and $g=3$. We take $u=2$ and verify that 
$(g-2u)L=(3-4)\cdot4=-4 \equiv 8 \pmod{12}$ and 
$8
\notin
[0,4]$.
The set $A$ in this case is : 
$A= \cup_{j=0}^{2} \{4j+1\}=\{1,5,9\}$ and the set $B=\{8,9,10\}$.

The set $Z= \{1, 5, 8, 9, 10\}$ is the defining set for the code $\mathcal{C}$. 
We further use MAGMA calculator \cite{MR1484478} to verify that there is no codeword of weight $4$ in $\mathcal{C} \setminus \mathcal{C}^\perp$, and moreover, $Q(d(\mathcal{C}))= 5 > d(\mathcal{C}) $. Thus, the induced quantum code is not pure and has parameters $[[12, 2, 5]]_{q=13}$.
\end{example}

\section{Dual-containing LRCs with unbounded length and quantum LRCs}\label{sec4}

In Construction \ref{C1} presented above, the code length is bounded above by the size of the underlying field, i.e.,  $n\le q-1$. However, in distributed data storage systems, it is desirable to have codes with a large length over relatively small fields, as this significantly reduces the computational cost and latency associated with finite field arithmetic. In this subsection, we will present cyclic LRCs of unbounded length that also possess the property of dual containment. We begin by solving the problem in its most general form in the following theorem.

\begin{theorem}\label{T1}
Let \(n\) be a positive integer and $q$ be a prime power such that \(\gcd(q,n)=1\). Let \(L\ge2\) be a positive divisor of \(n\) and $T\subseteq[1,L-1]$. Define a set
\[
S = \{\, jL + t \pmod n \;:\; 0 \le j \le M-1,\; t\in T \}, 
\qquad M = \tfrac{n}{L}.
\]
For each $a\in S$ let $C_a$ be the $q$-cyclotomic coset modulo $n$ containing $a$, and let \(H\) be the union of all $q$-cyclotomic cosets \(C_a\) modulo \(n\) with \(a \in S\), i.e.,
\[
H = \bigcup_{a\in S} C_a.
\]
Define the set
\[
\mathcal O = \{\, q^s t \pmod L \;:\; t\in T,\ s \ge 0 \}.
\]
Then
$
H \cap (-H) = \varnothing 
$ if and only if $
\mathcal O \cap \{-t \pmod{L} : t\in T\} = \varnothing.$

\end{theorem}

\begin{proof}

Suppose $H \cap (-H) \ne \varnothing$. This implies that there exists an element $x$ such that $x \in H$ and $-x \in H$.
By the definition of $H$, there exist $a, b \in S$ and non-negative integers $e, f$ such that $x \equiv a \cdot q^e \pmod{n}$ and $-x \equiv b \cdot q^f \pmod{n}$.
Substituting the first congruence into the second, we get $-(a \cdot q^e) \equiv b \cdot q^f \pmod{n}$. 

Without loss of generality, we may assume that $f > e$. 
This implies that $-a \equiv b \cdot q^{f-e} \pmod{n}$, which means that $-a$ is an element of the cyclotomic coset $C_b$. Thus, $-a \equiv b~q^s \pmod{n}$, where $s = f-e$.

We know that $a,b \in S$ implies $a=jL+t_a \pmod{n}$ and $b=kL+t_b\pmod{n}$ for some $0\le j,k \le M-1$. Substituting these values in the above equation, we get
$$(kL+t_a)q^s \equiv -(jL+t_b) \pmod{n}.$$
Using the fact that $L|n$, we reduce the above equation modulo $L$. This gives
\[
q^s t_a \equiv -t_b \pmod{L}.
\]
So $t_b\in \mathcal O\cap\{-t:t\in T\}$, which implies that $\mathcal O\cap\{-t:t\in T\}\neq\varnothing$.

 Conversely suppose that $\mathcal O\cap\{-t:t\in T\} \neq \varnothing$. Then there exist $t_a,t_b\in T$ and $s\ge0$ that satisfy
\[
q^s t_a \equiv -t_b \pmod L.
\]

Choose an integer \(j\in\{0,\dots,M-1\}\) and define
\[
a := jL + t_a \in S,\qquad b := j' L + t_b \in S,
\]
where \(j'\) is a non-negative integer. Then
\[
q^{s} a = q^{s} (jL+t_a) = L\big(j q^{s}\big) + q^{s}t_a.
\]

Since \(q^{s}t_a\equiv -t_b~\pmod L\), we have $q^s t_a+t_b = cL$ for some integer $c$. Now, let \(q^{s}a + b \equiv
L\big(j q^{s} + j'+c\big) ~(\mbox{mod}~n)\) for some \(j'\in\{0,\dots,M-1\}\). We choose $j'$ so that it satisfies
\[
j' \equiv -\,(j q^{s}+c) \pmod{M}.
\]
With this choice of $j'$, the integer \(j q^{s}+j'\) is a multiple of \(M\), and hence
\(L(j q^{s}+j')\) is a multiple of \(ML=n\). Therefore,
\[
q^{s}a + b \equiv 0 \pmod n, \quad \mbox{i. e.,} \quad q^{s}a \equiv -b \pmod n.
\]
Thus, \(-b\in H\).  Since \(-b\in -H\), we have
\(H\cap(-H)\neq\varnothing\).
Therefore, the two conditions are equivalent.

\end{proof}

Theorem \ref{T1} establishes a complete equivalence between the disjointness of the sets $H$ and $-H$. While exact, this criterion may be difficult to apply directly to obtain examples. Note that one can obtain dual-containing LRCs with $\delta = |T|+1$ using the above theorem. We next derive some restrictions on $L$ and obtain the set $H$ for particular values of $t$, which leads to the construction of LRCs with $\delta=2$.


\begin{theorem}\label{T2}

Let \(n = q^m - 1\) for some positive integer \(m\) and prime power $q$, let \(L\ge2\) be a divisor of \(n\), and let \(t < L\) be a positive integer such that \(\gcd(t,L)=1\).  
Define
\[
S_t = \{\ jL + t \pmod n \;:\; 0 \le j \le M-1 \}, 
\qquad M = \tfrac{n}{L},
\]
and let \(H = \bigcup_{a \in S_t} C_a\) be the union of all \(q\)-cyclotomic cosets modulo \(n\) with \(a \in S_t\).

If the multiplicative order \(o_L(q)\) of \(q\) modulo \(L\) is odd, then
$H \cap (-H) = \varnothing.
$ Moreover, the cardinality of \(H\) is given by
\[
|H| = M \cdot o_L(q).
\]  
\end{theorem}

\begin{proof}
By Theorem~\ref{T1}, the condition \(H \cap (-H) = \varnothing\) holds if and only if
\[
\mathcal O_t \cap \{-t \pmod{L}\} = \varnothing,
\qquad 
\text{where } 
\mathcal O_t = \{\, q^s t \pmod{L} : s \ge 0 \}.
\]
Since \(\gcd(t,L)=1\), multiplication by \(t^{-1}\) modulo \(L\) gives
\[
t^{-1} \mathcal O_t = \{\, q^s \pmod L : s \ge 0 \}.
\]
If \(o_L(q)\) is odd and $L>2$, then \(-1 \not\equiv q^s \pmod L\) for any \(s\). For if \(-1 \equiv q^s \pmod L\) for some $s$, then \(1 \equiv q^{2s} \pmod L\), implying that \(o_L(q) \mid 2s \), which gives  \(o_L(q)\mid s \), as  \(o_L(q)\) is odd. So  \(1 \equiv q^s \pmod L\), which is a contradiction. Hence
\(-1 \notin t^{-1} \mathcal{O}_t\), which implies that \(-t \notin \mathcal{O}_t\).   
By Theorem~\ref{T1}, we get \(H \cap (-H) = \varnothing.\)

Now, we determine $|H|$.
Let \(x \in H\). Then \(x = a q^b \pmod n\) for some \(a = jL + t \in S_t\) and $b\ge 0$.
Reducing modulo \(L\) gives
\[
x \equiv t q^b \pmod L.
\]
So the residue of \(x\) modulo \(L\) lies in the orbit
\(\mathcal{O}_t = \{\, q^s t \pmod{L} : s \ge 0 \,\}\).
Hence,
\[
H \subseteq \bigcup_{t' \in \mathcal{O}_t} S_{t'},
\qquad
\text{where } S_{t'} = \{\, jL + t' \pmod n : 0 \le j \le M-1 \}.
\]

Conversely, let \(t' \in \mathcal{O}_t\).
Then \(t' \equiv t q^s \pmod L\) for some integer \(s \ge 0\).
For each \(j \in \{0, 1, \ldots, M-1\}\), we have
\[ aq^s=
(jL + t) q^s= (jq^s+k)L+t' \text{ for some $k$. } \]

Let $j'=jq^s+k \pmod M$. Then there exists an integer $w$, such that $jq^s+k=j'+wM.$

Substituting back, we get $aq^s= (j'+wM)L+t'$. Now, using $n=ML$, we get

\[aq^s = j'L + t' \pmod n \in S_{t'}.\]

Since $n=q^m-1=ML$, we have \(\gcd(q, M) = 1\). The map $j\rightarrow j' \pmod M$ is a bijection of the index set $\{0,1,2, \cdots,M-1\}$. As $j$ runs through all possible values, so does $j'$. Therefore, all elements of \(S_{t'}\) appear among the residues of \(a q^s \pmod n\) as \(a\) varies in \(S_t\).
Hence \(S_{t'} \subseteq H\), giving
\[
H = \bigcup_{t' \in \mathcal{O}_t} S_{t'}.
\]

Each \(S_{t'}\) consists of \(M\) distinct elements modulo \(n\),
and for distinct residues \(t'_1, t'_2 \in \mathcal{O}_t\),
the sets \(S_{t'_1}\) and \(S_{t'_2}\) are disjoint,
since $t'_1$ and $t'_2$ differ modulo \(L\).
Therefore,
\[
|H| = \sum_{t' \in \mathcal{O}_t} |S_{t'}| = M \cdot |\mathcal{O}_t| = M\cdot o_L(q).
\] \end{proof}

The above theorem leads to the following construction. 

\begin{construction}[LRC with $\delta =2$]\label{c2}
Let \( n = q^{m} - 1 \) for some positive integer \( m \) and a prime power \( q \), and let \( L = r+1 \ge 3 \) be a divisor of \( n \). Let \( \gamma \) be a primitive \( n \)-th root of unity in the extension field \( \mathbb{F}_{q^{m}} \). Define the set
\[
R = \{ \gamma^{i} : i \in H \},
\]
where the index set \( H \) is as defined in Theorem~\ref{T2}. Let \( \mathcal{C} \) be the cyclic code whose generator polynomial has \( R \) as its set of roots.

\end{construction}

\begin{theorem}\label{thm6.2}
   The code $\mathcal{C}$, given in Construction \ref{c2} above, is an $[n=ML,M(L-o_L(q), \ge 2]$ dual-containing LRC with locality $r$ over $\mathbb{F}_q$. Equivalently, there exists an $(r,\delta=2)$-qLRC with parameters $[[ML,M(L-2o_L(q)), \ge 2]]_q$. 
\end{theorem}
\begin{proof}
    Note that the defining set of code $\mathcal{C}$  is $H$, which contains the set $S_t$, therefore by Theorem \ref{setA}, $\mathcal{C}$ is LRC with locality $r$. Also, $H \cap -H=\varnothing$, implies that the code $\mathcal{C}$ is dual-containing. Finally, $\deg g(x)=|R|=M.o_L(q)$, therefore the dimension $k$ of $\mathcal{C}$ is given by $k=n-\deg g(x) = M(L-o_L(q)).$

    The existence of $(r,\delta=2)$-qLRC with parameters $[[ML,M(L-2o_L(q)), \ge 2]]_q$ directly follows from Theorem \ref{CqLRC}.
\end{proof}

\begin{example}
    Let
$
q=2,\quad m=6,\quad n=2^{6}-1=63,\quad L=r+1=7,\quad M=\frac{n}{L}=9.
$
Let $\gamma$ be a primitive $63^{rd}$ root of unity in an extension of $\mathbb{F}_2$. Take $t=1$, and define 
\[
S_1 = \{\ 7j + 1 \pmod n \;:\; 0 \le j \le 8 \}=\{1,8,15,22,29,36,43,50,57\}. 
\]

We obtain \[
\mathcal O=\{2^{s}\cdot 1\pmod{7} : s\ge0\}=\{1,2,4\}.
\]

The set $H= \bigcup_{t'\in \mathcal{O}_t} S_{t'}=\{\ 7j + t' \pmod n \;:\; t' \in \mathcal{O},  0 \le j \le 8 \}$.
Define $$g(x)= \prod_{i \in H}(x-\gamma^i)\in \mathbb{F}_2[x].$$
The cyclic code $\mathcal{C}$ with the generator polynomial $g(x)$ is an $[n=63,k=36, d=3]$ dual-containing LRC with locality $r=6$ over $\mathbb{F}_2$. The corresponding quantum code has parameters $[[63,9,3]]_{q=2}$ with locality $6$.
\end{example}

We further restrict the conditions on parameters to easily construct such a set $H$. The following corollary is a direct implication of Theorem \ref{T2}. 

\begin{corollary}
Let \(m\) be an odd positive integer and \(n = q^m - 1\). Let \(L\) be a divisor of \(n\). Define the set
\[
S = \{jL + 1 \mid 1 \leq jL + 1 \leq n\},
\]
and let \(H\) be the union of all $q$-cyclotomic cosets \(C_a\) modulo \(n\) with  \(a \in S\). Then, \(H \cap (-H) = \emptyset\).
\end{corollary}
\begin{proof} It suffices to show that $o_L(q)$ is odd. Now, as  $q^m \equiv 1 
 \pmod{n}$ and $L\mid n$, the equation, $q^m \equiv 1 
 \pmod{L}$ holds. This implies that $o_L(q)\mid m$. Since $m$ is odd, $o_L(q)$ is odd.
    
\end{proof}

\begin{example}
    Let $q=2, m=9,$  $n=q^{m}-1=2^{9}-1=511,$
and $L=73$. Then $M=\frac{n}{L}=\frac{511}{73}=7$.

Let \(\gamma\) be a primitive \(511\)-st root of unity in the extension field \(\mathbb{F}_{2^9}\) of \(\mathbb{F}_2\). Take \(t=1\), and define
\[
S_1=\{\;73j+1 \pmod{511}\;:\;0\le j\le 6\;\}=\{1,74,147,220,293,366,439\}.
\]
We obtain
\[
\mathcal{O}=\{2^{s}\cdot 1 \pmod{73} : s\ge0\}
=\{1,2,4,8,16,32,64,55,37\}.
\]
The set $\mathcal{O}$ has cardinality $9$ since \(o_{73}(2)=9\).

Let
\[
H=\bigcup_{a \in S_1} C_a=\bigcup_{t'\in\mathcal{O}} S_{t'}
=\bigcup_{t'\in\mathcal{O}}\{\;73j+t'\pmod{511}\;:\;0\le j\le6\;\}.
\]
Define the polynomial
\[
g(x)=\prod_{i\in H} (x-\gamma^i)\in\mathbb{F}_2[x].
\]
The set \(H\) has \(|\mathcal{O}|\cdot M=9\cdot7=63\) elements, so \(\deg g(x)=63\). Therefore the cyclic code \(\mathcal{C}\) with the generator polynomial \(g(x)\) is an \([511,511-63]=[511,448]\) dual-containing LRC over \(\mathbb{F}_2\) with locality $r=L-1=72$, and minimum distance 3. The corresponding quantum code has parameters $[[511,385, 3]]_{q=2}$ with locality $r=72$.

\end{example}

\begin{example}
Let
$q=3, m=5, n=q^{m}-1=3^{5}-1=242,$
and $L=11$. Then $M=\frac{n}{L}=\frac{242}{11}=22.$

Let \(\gamma\) be a primitive \(242\)-nd root of unity in an extension of \(\mathbb{F}_3\). Take \(t=1\), and define
\[
S_1=\{\;11j+1 \pmod{242}\;:\;0\le j\le 21\;\}.
\]
We compute \[
\mathcal{O}=\{3^{s}\cdot 1 \pmod{11} : s\ge0\}
=\{1,3,9,5,4\}.
\]
Thus, we have \(o_{11}(3)=5\).
Let
\[
H=\bigcup_{t'\in\mathcal{O}} S_{t'}
=\bigcup_{t'\in\{1,3,9,5,4\}}\{\;11j+t'\pmod{242}\;:\;0\le j\le21\;\}.
\]
Define the polynomial
\[
g(x)=\prod_{i\in H} (x-\gamma^i)\in\mathbb{F}_3[x].
\]
The set \(H\) has \(|\mathcal{O}|\cdot M=5\cdot22=110\) elements, so \(\deg g(x)=110\). Therefore the cyclic code \(\mathcal{C}\) with the generator polynomial \(g(x)\) is a \([242,132,5]\) dual-containing locally recoverable code over \(\mathbb{F}_3\) with locality $r=L-1=10$. The corresponding qLRC has parameters $[[242,22,5]]_{q=3}$ with locality $10$.

\end{example}

In the examples presented above, we notice that it is possible to construct such sets for a specific value of \( t \). However, when $|T| \ge2$, the problem becomes complicated, as it is challenging to find numbers that satisfy multiple conditions simultaneously. A simpler way to find multiple continuous segments for each multiple of \( L \), while maintaining the property that \(H\cap(-H)=\varnothing\), is given in the following theorem. Here, we assume \( q \equiv 1 \pmod L\). The authors in \cite{FANG2020101650} use a similar condition on $q$ to construct LRCs with unbounded length.

\begin{theorem}\label{T6}
Let \(n = q^{m} - 1\) for some integer \(m > 0\) and prime power \(q\).  
Let \(L \ge 2\) be a divisor of \(n\), and suppose that \(q \equiv 1 \pmod L\).  
Let \(M = n/L\) and \(T \subseteq \{1,2,\dots,\lceil L/2 \rceil - 1\}\).  
Define
\[
S = \{\, jL + t \pmod n \;:\; 0 \le j \le M-1,\; t \in T \}.
\]
For each \(a \in S\), let \(C_a\) denote the \(q\)-cyclotomic coset modulo \(n\) containing \(a\), and let
\[
H = \bigcup_{a \in S} C_a.
\]
Then \(H = S\), and consequently \(|H| = M|T| = \dfrac{n}{L}\,|T|\). Also, $H$ is inverse free, i.e. $H\cap(-H)=\varnothing$.
\end{theorem}

\begin{proof}
Let \(a = jL + t \in S\) with \(0 \le j \le M-1\) and \(t \in T\).  
For any integer \(i \ge 0\), consider
\[
x \equiv q^i a \pmod n,
\qquad 0 \le x \le n-1.
\]
Since \(q \equiv 1 \pmod L\), we have \(q^i \equiv 1 \pmod L\).  
Hence
\[
x \equiv q^i(jL + t) \equiv jL + t \equiv t \pmod L.
\]
Therefore \(x - t\) is divisible by \(L\). Write \(x - t = Lk\) for some integer \(k\).

Because \(0 \le x \le n-1 = ML - 1\) and \(0 \le t \le L - 1\), we get
\[
0 \le x - t < ML.
\]
Then \(k = \tfrac{x - t}{L}\) satisfies \(0 \le k \le M - 1\).  
It follows that
\[
x = kL + t, \qquad 0 \le k \le M - 1,\ t \in T,
\]
and hence \(x \in S\).

Thus, every element of the \(q\)-cyclotomic coset \(C_a\) lies in \(S\), i.e. \(C_a \subseteq S\).  
The reverse inclusion \(S \subseteq H\) is trivial from the definition of \(H\).  
Therefore \(H = S\), and the cardinality of $H$ is
\[
|H| = |S| = M|T| = \frac{n}{L}\,|T|.
\]
Finally we show \(H\cap(-H)=\varnothing\).  By Theorem~\ref{T1} we have the equivalence
\[
H\cap(-H)=\varnothing
\quad\Longleftrightarrow\quad
\mathcal O\cap\{-t\pmod{L}:t\in T\}=\varnothing,
\]
where \(\mathcal O=\{q^s t\pmod{L} : t\in T,\ s\ge0\}\).  Under the hypothesis \(q\equiv1\pmod L\)
the orbit is trivial: for every \(t\in T\) and every \(s\ge0\) we have \(q^s t\equiv t\pmod L\),
so \(\mathcal O=T\). Hence, the condition $
\mathcal O \cap \{-t \pmod{L} : t\in T\} = \varnothing$ of Theorem~\ref{T1} reduces to
\[
T \cap \{-t\pmod{L} : t\in T\} = \varnothing.
\]
Now, \(T\subseteq\{1,\dots,\lceil L/2\rceil-1\}\) ensures that for every
\(t\in T\) the residue \(-t\pmod{L} = L-t\) lies in \(\{\lfloor L/2\rfloor+1,\dots,L-1\}\),
which is disjoint from \(\{1,\dots,\lceil L/2\rceil-1\}\). Then it follows that 
\(T \cap \{-t\pmod{L} : t\in T\}=\varnothing\), and hence by Theorem~\ref{T1} we conclude that
\(H\cap(-H)=\varnothing\).
This completes the proof.
\end{proof}

Using the above theorem, we present a construction of a dual-containing $(r,\delta)$-LRC.

\begin{construction}\label{c3}
Let \(n = q^m - 1\) for some positive integer \(m\) and prime power $q$, let \(L=r+\delta-1\ge2\) be a divisor of \(n\) such that $q \equiv 1 \pmod L$, and $2\le \delta \le \lceil L/2\rceil $. Define a set 
\[
R = \{\, jL + t \pmod n \;:\; 0 \le j \le M-1,\; 1 \le t \le \delta-1 \,\}.
\]
Let \( \gamma \) be a primitive \( n \)-th root of unity in the field \( \mathbb{F}_{q^m} \). Define a polynomial 
\[
g(x) = \prod_{i \in R} (x - \gamma^i).
\]
Let $\mathcal{C}$ be a cyclic code with generator polynomial $g(x)$.
\end{construction}

\begin{theorem} \label{thm6.3}
      The linear code $\mathcal{C}$, given in Construction \ref{c3}, is an optimal $[n,Mr,\delta]$ dual-containing $(r,\delta)$-LRC over $\mathbb{F}_q$. Equivalently, there exist an $(r,\delta)$-qLRC with parameters $[[M(r+\delta-1),M(r-\delta+1), \ge \delta]]_q$. 
\end{theorem}

\begin{proof}
It is clear from Theorem \ref{T6} that the set $R$ is closed with respect to $q$-cycotomic cosets modulo $n$, and hence $g(x) \in \mathbb{F}_q$. This implies that the code $\mathcal{C}$ is over $\mathbb{F}_q$. This also implies that the defining set for code $\mathcal{C}$ can be directly obtained from $R$ itself. Therefore, by Theorem \ref{setA}, $\mathcal{C}$ is an $(r,\delta)$-LRC.

The optimality of $\mathcal{C}$ can be proved in a similar manner as in Theorem \ref{6.Thm}.  Also, note that the set $R$ is contained in the set $H$ given in Theorem \ref{T6}. Moreover, we have $\delta \le \lceil \frac{L}{2}\rceil$. Combining all together, and by Theorem \ref{T6}, we conclude that $\mathcal{C}$ is a dual-containing $(r,\delta)$-LRC over $\mathbb{F}_q$. The parameters of the corresponding qLRC can be obtained from Theorem \ref{CqLRC}.
\end{proof}

\begin{example}
   Let $q=11,  m=2, n=q^{m}-1=11^{2}-1=120$, and
\(L=r+\delta-1=10\) such that $\delta \le \lceil L/2\rceil=5$. Then we get $M=\frac{n}{L}=\frac{120}{10}=12.$
For this example, we fix $\delta=4$ and thus we get $r=7$.
Note that we have \(q=11\equiv1\pmod{L=10}\).
The set $R$ is given as 
\[
R = \{\, 10j + t \pmod n \;:\; 0 \le j \le 11,\; 1 \le t \le 3 \,\}.
\]
The $11$-cyclotomic coset of any element $a \in R$ will contain at most two elements, and is given by $C_a=\{a, 11a \pmod {120} \}$. Moreover, since \(q\equiv1\pmod L\), we get \(H=R\) and
$|H|=|R|=12\cdot 3 = 36.$

Now, let $\gamma$ be a primitive $120^{th}$ primitive root of unity in $\mathbb{F}_{121}$, and let $\mathcal{C}$ be a cyclic code with the generator polynomial having roots $\gamma^i: i \in R$. Then $\mathcal{C}$ is a $[120,84, 4 ]$ dual-containing $(r=7, \delta=4)$-LRC over $\mathbb{F}_{11}$. The corresponding quantum code is a distance-optimal pure $(r=7, \delta=4)$-qLRC with parameters $[[120,48,4]]_{q=11}$.

\end{example}

\section{Conclusion}
In this work, we have studied some constructions of qLRCs from classical cyclic $(r,\delta)$-LRCs. By determining which defining sets of cyclic LRCs yield dual-containing conditions, we developed a method for constructing CSS-type qLRCs. We presented three explicit families of $(r,\delta)$-qLRCs, two of which attain the quantum Singleton-like bound whenever the codes are pure. Moreover, Constructions~\ref{c2} and~\ref{c3} admit code lengths that are unbounded with respect to the field size, which is desirable for practical applications. 
As a direction for future work, our methods based on Theorems~\ref{T2} and~\ref{T6} primarily address the cases where \(\delta = 2\), and \(\delta \leq \lceil\frac{L}{2} \rceil\) with \(q \equiv 1 \pmod{L}\), respectively. However, extending these constructions to cover arbitrary values of \(\delta\) without the condition \(q \equiv 1 \pmod{L}\), for unbounded lengths with respect to field sizes, remains an open and challenging problem. In particular, finding examples with \(|T| \geq 2\) for Theorem~\ref{T1} presents an interesting challenge.

\section*{Acknowledgment}
The first author would like to thank the University Grants Commission, India, for financial support. 





\bibliography{sn-bibliography}
\end{document}